\newif\ifprocs
\procstrue
\procsfalse 

\ifprocs
\documentclass[oribibl]{myllncs}
\else
\documentclass[11pt]{article}
\usepackage{amsthm}
\usepackage{fullpage}
\fi 

\usepackage{graphicx}

\usepackage{verbatim} 
\usepackage[noend]{algorithmic}
\usepackage{algorithm}
\usepackage{amsmath,amssymb,amsfonts}

\usepackage{color}
\usepackage{xspace}

\usepackage{ifpdf}
\ifpdf    
\usepackage{hyperref}
\else    
\usepackage[hypertex]{hyperref}
\fi

\ifprocs

\spnewtheorem{observation}{Observation}{\bfseries}{\itshape}
\else
\newtheorem{theorem}{Theorem}[section]
\newtheorem{corollary}[theorem]{Corollary}
\newtheorem{lemma}[theorem]{Lemma}
\newtheorem{observation}[theorem]{Observation}
\newtheorem{definition}[theorem]{Definition}

\newtheorem{question}[theorem]{Question}
\newtheorem{remark}[theorem]{Remark}

\fi

\ifprocs
\newcommand{\proofof}[1]{#1}
\else
\newcommand{\proofof}[1]{Proof #1}
\fi

\def\ala{\`a la\xspace}

\newcommand{\calf}{\mathcal{F}}
\newcommand{\redtw}{\textsc{ReduceGraphTW}\xspace}
\newcommand{\naive}{\textsc{ReduceGraphNaive}\xspace}
\newcommand{\separator}{\textsc{Separator}}

\newcommand{\R}{{\mathbb R}}

\DeclareMathOperator{\trees}{\sf Trees}
\DeclareMathOperator{\planar}{\sf Planar}
\DeclareMathOperator{\tw}{\sf Treewidth}

\def\compactify{\itemsep=0pt \topsep=0pt \partopsep=0pt \parsep=0pt}

\title{Preserving Terminal Distances using Minors%
\thanks{A preliminary version of this paper appeared in Proceedings of ICALP 2012.
This work was supported in part by The Israel Science Foundation
(grant \#452/08), by a US-Israel BSF grant \#2010418, 
and by the Citi Foundation.}
}

\ifprocs
\author{Robert Krauthgamer
\and 
Tamar Zondiner
}
\institute{Weizmann Institute of Science, Rehovot, Israel.\\
\email{\{robert.krauthgamer,tamar.zondiner\}@weizmann.ac.il}
}

\else
\author{Robert Krauthgamer
\qquad\qquad Tamar Zondiner
\\ Weizmann Institute of Science
\\ \texttt{\{robert.krauthgamer,tamar.zondiner\}@weizmann.ac.il}
}
\fi

\begin{document}

\maketitle

\begin{abstract}
We introduce the following notion of compressing an undirected graph $G$ 
with (nonnegative) edge-lengths and terminal vertices $R\subseteq V(G)$. 
A \emph{distance-preserving minor} is a minor $G'$ (of $G$) 
with possibly different edge-lengths, such that $R\subseteq V(G')$ and 
the shortest-path distance between every pair of terminals 
is exactly the same in $G$ and in $G'$.
We ask: what is the smallest $f^*(k)$ such that every graph $G$ with $k=|R|$ terminals 
admits a distance-preserving minor $G'$ with at most $f^*(k)$ vertices?

Simple analysis shows that $f^*(k)\le O(k^4)$.
Our main result proves that $f^*(k)\ge \Omega(k^2)$,
significantly improving over the trivial $f^*(k)\ge k$.
Our lower bound holds even for planar graphs $G$,
in contrast to graphs $G$ of constant treewidth,
for which we prove that $O(k)$ vertices suffice. 
\end{abstract}

\section{Introduction}
A \emph{graph compression} of a graph $G$ is 
a small graph $G^*$ that preserves certain features (quantities) of $G$,
such as distances or cut values.
This basic concept was introduced by Feder and Motwani \cite{FM95},
although their definition was slightly different technically.
(They require that $G^*$ has fewer edges than $G$,
and that each graph can be quickly computed from the other one.)
Our paper is concerned with preserving the selected features of $G$ 
\emph{exactly} (i.e., lossless compression), but in general 
we may also allow the features to be preserved approximately.

The algorithmic utility of graph compression is readily apparent -- 
the compressed graph $G^*$ may be computed as a preprocessing step,
and then further processing is performed on it (instead of on $G$)
with lower runtime and/or memory requirement.
This approach is clearly beneficial 
when the compression can be computed very efficiently, say in linear time,
in which case it may be performed on the fly,
but it is useful also when some computations are to be performed (repeatedly) 
on a machine with limited resources such as a smartphone,
while the preprocessing can be executed in advance on much more powerful 
machines.

For many features, graph compression was already studied and many results are 
known.
For instance, a \emph{$k$-spanner} of $G$ is a subgraph $G^*$ 
in which all pairwise distances approximate those in $G$ within a factor of $k$ 
\cite{PS89}.
Another example, closer in spirit to our own, is a \emph{sourcewise distance preserver} of $G$ with respect to a set of vertices $R\subseteq V(G)$; this is a subgraph $G^*$ of $G$ that preserves (exactly) the distances in $G$ for all pairs of vertices in $R$ \cite{CE06}. 
We defer the discussion of further examples and related notions 
to Section \ref{sec:related}, and here point out only two phenomena:
First, it is common to require $G^*$ to be structurally similar to $G$ 
(e.g., a spanner is a subgraph of $G$),
and second, sometimes only the features of a subset $R$ need to be preserved
(e.g., distances between vertices of $R$). 

\medskip
We consider the problem of compressing a graph so as to maintain
the shortest-path distances among a set $R$ of required vertices.
From now on, the required vertices will be called \emph{terminals}.

\begin{definition} \label{defn:dpm}
Let $G$ be a graph with edge lengths $\ell:E(G)\to\R_+$
and a set of terminals $R\subseteq V(G)$.
A \emph{distance-preserving minor} (of $G$ with respect to \ $R$) 
is a graph $G'$ with edge lengths $\ell':E(G')\to\R_+$ satisfying:
\begin{enumerate} \compactify
\item $G'$ is a minor of $G$; and
\item $d_{G'}(u,v)=d_G(u,v)$ for all $u,v\in R$.
\end{enumerate}
\end{definition}
Here and throughout, $d_H$ denotes the shortest-path distance in a graph $H$.
It also goes without saying that the terminals $R$ must survive the minor operations
(they are not removed, but might be merged with non-terminals, due to edge contractions),
and thus $d_{G'}(u,v)$ is well-defined; in particular, $R\subseteq V(G')$. 
For illustration, suppose $G$ is a path of $n$ unit-length edges
and the terminals are the path's endpoints; then by contracting all the edges, 
we can obtain $G'$ that is a single edge of length $n$.

The above definition basically asks for 
a minor $G'$ that preserves all terminal distances exactly. 
The minor requirement is a common method to induce structural similarity 
between $G'$ and $G$, and in general excludes the trivial solution 
of a complete graph on the vertex set $R$ (with appropriate edge lengths).
The above definition may be viewed as a conceptual contribution of our paper,
and indeed our main motivation is its mathematical elegance,
but for completeness we also present potential algorithmic applications
in section \ref{sec:applications}.

We raise the following question,
which to the best of our knowledge was not studied before.
Its main point is to bound the size of $G'$ independently of the size of $G$.
\begin{question} \label{Q1}
What is the smallest $f^*(k)$, such that for every graph $G$ with 
$k$ terminals,
there is a distance-preserving minor $G'$ with at most $f^*(k)$ vertices?
\end{question}

Before describing our results, let us provide a few initial observations,
which may well be folklore or appear implicitly in literature.
There is a \emph{naive algorithm} which constructs $G'$ from $G$ 
by two simple steps
(Algorithm \ref{alg:naive} in Section \ref{sec:trivial}):
\begin{enumerate} \compactify
\item[(1)]
Remove all vertices and edges in $G$ that do not
participate in any shortest-path between terminals.
\item[(2)]
Repeat while the graph contains a non-terminal $v$ of degree two:
merge $v$ with one of its neighbors (by contracting the appropriate edge),
thereby replacing the $2$-path $w_1-v-w_2$ with a single edge 
$(w_1,w_2)$ of the same length as the $2$-path.

\end{enumerate}
It is straightforward to see that these steps reduce the number of 
non-terminals without affecting terminal distances,
and a simple analysis proves that this algorithm always produces 
a minor with $O(k^4)$ vertices and edges (and runs in polynomial time).
It follows that $f^*(k)$ exists, and moreover
$$f^*(k) \le O(k^4).$$
Furthermore, if $G$ is a tree then $G'$ has at most $2k-2$ vertices,
and this last bound is in fact tight (attained by a complete binary tree)
whenever $k$ is a power of $2$.
We are not aware of explicit references for these analyses, 
and thus review them in Section \ref{sec:trivial}.

\subsection{Our Results} \label{sec:results}

Our first and main result directly addresses Question \ref{Q1},
by providing the lower bound $f^*(k)\ge \Omega(k^2)$.
The proof uses only simple planar graphs,
leading us to study the restriction of $f^*(k)$ to specific graph families, 
defined as follows.%
\footnote{We use $(V,E,\ell)$ to denote a graph with vertex set $V$, 
edge set $E$, and edge lengths $\ell:E\to\R_+$.
As usual, the definition of a family $\calf$ of graphs refers only to the
vertices and edges, and is irrespective of the edge lengths.
}

\begin{definition} \label{defn:fstar}
For a family $\calf$ of graphs, define $f^*(k, \calf)$ as the minimum value 
such that every graph $G=(V,E,\ell)\in \calf$ with $k$ terminals
admits a distance-preserving minor $G'$ with at most $f^*(k, \calf)$ vertices. 
\end{definition} 

\begin{theorem} \label{thm:main}
Let $\planar$ be the family of all planar graphs. Then 
$$f^*(k)\ge f^*(k, \planar)\ge \Omega(k^2).$$
\end{theorem}

Our proof of this lower bound uses a two-dimensional grid graph, 
which has super-constant treewidth.
This stands in contrast to graphs of treewidth $1$, 
because we already mentioned that
$$f^*(k,\trees) \le 2k-2,$$
where $\trees$ is the family of a all tree graphs.
It is thus natural to ask whether bounded-treewidth graphs 
behave like trees, for which $f^* \le O(k)$,
or like planar graphs, for which $f^*\ge \Omega(k^2)$. 
We answer this question as follows.

\begin{theorem} \label{thm:tw}
Let $\tw(p)$ be the family of all graphs with treewidth at most $p$. 
Then for all $k\ge p$,
$$ \Omega(pk) \le f^*(k, \tw(p))\le O(p^3k).$$
\end{theorem}

We summarize our results together with some initial observations
in the table below.

\begin{center}
\begin{tabular}{|l|cc|l|}\hline
Graph Family $\calf$ 
  & \multicolumn{2}{|c|}{Bounds on $f^*(k,\calf)$} &  \\
\hline
Trees & $=2k-2$ &
  & Theorems \ref{thm:lbtree}, \ref{thm:ubtree}\\
Treewidth $p$ & $\Omega(pk)$ & $O(p^3k)$ &  Theorem \ref{thm:tw}\\
Planar Graphs & $\Omega(k^2)$ & $O(k^4)$ &  Theorems \ref{thm:main}, \ref{thm:k4}\\
All Graphs & $\Omega(k^2)$ & $O(k^4)$ &  Theorems \ref{thm:main}, \ref{thm:k4}\\
\hline
\end{tabular}
\end{center}

All our upper bounds are algorithmic and run in polynomial time. 
In fact, they can be achieved using the naive algorithm (Algorithm \ref{alg:naive} in Section \ref{sec:trivial}).

\subsection{Related Work} \label{sec:related}

Coppersmith and Elkin \cite{CE06} studied a problem similar to ours, except that they seek subgraphs with few edges (rather than minors). Among other things, they prove that for every weighted graph $G=(V,E)$ and every set of $k=O(|V|^{{1}/{4}})$ terminals (sources), there exists a weighted subgraph $G'=(V,E')$, called \emph{a source-wise preserver}, that preserves terminal distances exactly
and has $|E'|\le O(|V|)$ edges.
They also show a nearly-matching lower bound on $|E'|$.
Dor, Halperin and Zwick \cite{DHZ00} similarly asked for a graph with few edges, though not necessarily a subgraph or a minor, that preserves all distances.
Woodruff \cite{Woodruff06} combined their notion of \emph{emulators} with Coppersmith and Elkin's \emph{source-wise preservers}, and studied the size of arbitrary graphs preserving only distances between given sets of terminals (sources) in the given graph $G$.

Some compressions preserve cuts and flows in a given graph $G$ rather than distances. A Gomory-Hu tree \cite{GH61} is a weighted tree that preserves all $st$-cuts in $G$ (or just between terminal pairs).
A so-called mimicking network preserves all flows and cuts between 
subsets of the terminals in $G$ \cite{HKNR98}.

Terminal distances can also be approximated instead of preserved exactly. 
In fact, allowing a constant factor approximation may be sufficient to obtain a compression $G^*$ without any non-terminals. Gupta \cite{Gupta01} introduced this problem and proved that for every weighted tree $T$ and set of terminals, there exists a weighted tree $T'$ without the non-terminals that approximates all terminal distances within a factor of $8$. 
It was later observed that this $T'$ is in fact a minor of $T$ \cite{CGNRS06}, 
and that the factor $8$ is tight \cite{CXKR06}.
Basu and Gupta \cite{BG08} claimed that a constant approximation factor exists for weighted outerplanar graphs as well. 
It remains an open problem whether the constant factor approximation extends 
also to planar graphs (or excluded-minor graphs in general). Englert et al. \cite{EGKRTT10} proved a randomized version of this problem for all excluded-minor graph families, with an expected approximation factor depending only on the size of the excluded minor. 

The relevant information (features) in a graph can also be maintained 
by a data structure that is not necessarily graphs. 
A notable example is Distance Oracles -- low-space data structures that can answer distance queries (often approximately) in constant time \cite{TZ05}. These structures adhere to our main requirement of ``compression'' 
and are designed to answer queries very quickly.
However, they might lose properties that are natural in graphs,
such as the triangle inequality or the similarity of a minor to the given graph,
which may be useful for further processing of the graph.

\subsection{Potential Applications} \label{sec:applications}

Our first example application is in the context of algorithms
dealing with graph distances.
Often, algorithms that are applicable to an input graph $G$ 
are applicable also to a minor of it $G'$
(e.g., algorithms for planar graphs).
Consider for instance the Traveling Salesman Problem (TSP),
which is known to admit a QPTAS in excluded-minor graphs \cite{GS02}
(and PTAS in planar graphs \cite{Klein08}),
even if the input contains a set of \emph{clients}
(a subset of the vertices that must be visited by the tour).
Suppose now that the clients change daily,
but they can only come from a fixed and relatively small set $R\subset V(G)$ 
of potential clients.
Obviously, once a distance-preserving minor $G'$ of $G$ is computed,
the QPTAS can be applied on a daily basis to the small graph $G'$ 
(instead of to $G$).
Notice how important it is to preserve all terminal distances exactly
using $G'$ that is a minor of $G$ (a complete graph on vertex set $R$ 
would not work, because we do not have a QPTAS for it).

Our second example application is in the field of metric embeddings. 
Consider a known embedding, such as the embedding of a bounded-genus graph $G$
into a distribution over planar graphs \cite{IS07}. 
Suppose we want to use this embedding, but we only care about 
a small subset of the vertices $R\subset V(G)$.
We can compute a distance-preserving minor $G'$ (and thus with same genus)
that has at most $f^*(|R|)$ vertices,
and then apply the said embedding to the small graph $G'$ (instead of to $G$).
The resulting planar graphs will all have $f^*(|R|)$ vertices,
independently of $|V(G)|$.
In (other) cases where the embedding's distortion depends on $|V(G)|$,
this approach may even yield improved distortion bounds,
such as replacing $O(\log|V(G)|)$ terms with $O(\log |R|)$.

\section{Review of Straightforward Analyses} \label{sec:trivial}
As described in the introduction, a naive way to create a minor $G'$ of $G$ preserving terminal distances is to perform the steps described in \naive, depicted below as Algorithm \ref{alg:naive}. In this section we show that for general graphs $G$, the returned minor has at most $O(k^4)$ vertices, and for trees it has at most $2k-2$ vertices. 
\algsetup{indent=2em}
\begin{algorithm}                      
\caption{\naive(graph $G$, required vertices $R$)}          
\label{alg:naive}                           
\begin{algorithmic}[1]

\STATE Remove non-terminals and edges that do not participate in any terminal-to-terminal shortest-path.

\WHILE{there exists a non-terminal $v$ incident to only two edges $(v,u)$ and $(v,w)$}
\STATE contract the edge $(u,v)$,
\STATE set the length of edge $(u,w)$ to be $d_G(u,w)$.
\ENDWHILE
\end{algorithmic}
\end{algorithm}

It is easy to see that $G'$ is a distance-preserving minor of $G$ with respect to $R$.

\subsection{$f^*(k)\le O(k^4)$ for General Graphs}

\begin{theorem} \label{thm:k4}
For every graph $G$ and set $R\subseteq V$ of $k$ terminals, the output $G'$ of $\naive(G, R)$ is a distance-preserving minor of $G$ with at most $O(k^4)$ vertices. 
In particular, $f^*(k)\le O(k^4)$.
\end{theorem}
\begin{proof}
We need the following lemma, whose proof is sketched below. A detailed proof is shown in \cite[Lemma 7.5]{CE06}, where it is used to bound the number of edges in the graph $G'$ after only performing on a graph the edge-removals in line 1 of \naive. 

\begin{lemma}
Let $G$ be a graph, and suppose that ties between shortest paths 
(conneting the same pair of terminals) are broken in a consistent way.
Then every two distinct shortest paths between terminals in $G$, 
denoted $\Pi$ and $\Pi '$, branch in at most two vertices, 
i.e., there at most two vertices $v\in V(\Pi)\cap V(\Pi ')$ such that $\operatorname{succ}_\Pi(v)\notin V(\Pi ')$ or $\operatorname{pred}_\Pi(v)\notin V(\Pi ')$. 
\end{lemma}
\begin{proof}[\proofof{Sketch}]
Suppose that ties between two shortest paths are broken in a consistent way (by using extremely small perturbations to edge-weights when computing the shortest paths). Let $v_1$ and $v_2$ be the first and last vertices on the path $\Pi$ such that $v_1,v_2\in V(\Pi)\cap V(\Pi ')$. Then the path between $v_1$ and $v_2$ is shared in both the shortest path $\Pi$ and $\Pi '$, and contains no additional branching vertices. 
\ifprocs\qed\fi
\end{proof}

Every non-terminal $v\in V'\setminus R$ has degree greater or equal to 3, hence it is a branching vertex. Every pair of shortest paths contributes at most 2 branching vertices to $G'$. There are $O(k^4)$ such pairs, and therefore $O(k^4)$ vertices in $V'$. Since $G'$ is also a distance-preserving minor of $G$ with respect to $R$, this completes the proof of Theorem \ref{thm:k4}.
\ifprocs\qed\fi
\end{proof}

It is interesting to note that $G'$ is relatively sparse, having only $O(k^4)$ edges as well as vertices. It is easy to see that any branching vertex between two paths $\Pi_1, \Pi_2$ that participates also in the path $\Pi(t_1, t_2)$ is also a branching vertex between one of these paths $\Pi_i$ and $\Pi(t_1, t_2)$ itself. Therefore, at most $O(k^2)$ vertices, and hence also edges, appear on the contracted path between $t_1$ and $t_2$ in $G'$, and $G'$ has at most $O(k^4)$ edges overall.

\subsection{$f^*(k, \trees)=2k-2$}

\begin{theorem} \label{thm:ubtree}
For every tree $G$ and set $R\subseteq V$ of $k$ terminals, the output $G'$ of $\naive(G, R)$ is a distance-preserving minor of $G$ with at most $2k-2$ vertices. 
In particular, $f^*(k, \trees)\le 2k-2$.
\end{theorem}

\begin{proof}
Every non-terminal $v\in V'\setminus R$ has degree greater or equal to 3. Let $s$ denote the number of non-terminals in the tree $G'$. Then 
$$\sum_{v\in V'} \deg_{G'}(v) \ge k + 3s.$$
Since $G'$ is a tree, the sum of its degrees also equals $2(k+s)-2$, hence 
$2(k+s)-2\ge k+3s$, and $s\le k-2$, proving the theorem.
\ifprocs\qed\fi
\end{proof}

This bound is exactly tight. We sketch the proof of the following theorem.
\begin{theorem} \label{thm:lbtree}
For every $i\in \mathbb{N}$ there exists a tree $G$ and $k=2^i$ terminals $R\subseteq V$ such that every distance-preserving minor $G'$ of $G$ with respect to $R$ has $|V'|\ge 2k-2$.
In particular, $f^*(k, \trees)\geq 2k-2$ for $k=2^i$.
\end{theorem}
\begin{proof}[\proofof{Sketch}]
Consider the complete binary tree $G$ of depth $i$ with unit edge-lengths. Let the $2^i$ leaves of the tree be the terminals $R$. We use induction on $i$ to prove that for the complete binary tree with level $i$, the only edge contraction (and indeed the only minor operation) allowed is the contraction of an edge between the root and one of its children. In the tree with depth 1 this is clearly true. Let $T$ be the complete binary tree with depth $i+1$, and $T_1$, $T_2$ be its two $i$-depth subtrees. Any minor of $T$ does not combine the minors for $T_1$ and $T_2$, since paths between $v, u\in V(T_i)$ are always shorter than paths between $v\in V(T_1)$ and $u\in V(T_2)$. The induction hypothesis therefore rules out edge-contractions not involving the roots of $T_1$ and $T_2$. Pairwise distances between terminals inside and between the trees $T_1$ and $T_2$ dictate that, again, the only possible edge-contraction in $T$ is that of (without loss of generality) the edge $(root(T_1), root(T))$, reducing the number of vertices to $2k-2$.
\ifprocs\qed\fi
\end{proof} 

\section{A Lower Bound of $\Omega(k^2)$}

In this section we prove Theorem \ref{thm:main} using an even stronger assertion: there exist planar graphs $G$ such that every distance-preserving \emph{planar graph} $H$ (a planar graph with $R\subseteq V(H)$ that preserves terminal distances) has $|V(H)|\ge \Omega(k^2)$. Since any minor $G'$ of $G$ is planar, Theorem \ref{thm:main} follows.

Our proof uses a $k\times k$ grid graph with $k$ terminals, whose edge-lengths are
chosen so that terminal distances are essentially ``linearly independent'' of one another. 
We use this independence to prove that no distance-preserving minor $G'$ 
can have a small vertex-separator. 
Since $G'$ is planar, we can apply the planar separator theorem \cite{LT79},
and obtain the desired lower bound.

\begin{theorem} \label{thm:grid}
For every $k\in \mathbb{N}$ there exists a planar graph $G=(V,E,\ell)$ (in particular, the $k\times k$ grid) and $k$ terminals $R\subseteq V$, such that every distance-preserving \emph{planar graph} $G'=(V',E', \ell ')$ has $\Omega(k^2)$ vertices. 
In particular, $f^*(k, \mathrm{Planar})\ge \Omega(k^2)$. 
\end{theorem}

\begin{proof}
For simplicity we shall assume that $k$ is even. Consider a grid graph $G$ of size $k\times k$ with vertices $(x,y)$ for $x,y\in [0,k-1]$. Let the length function $\ell$ be such that the length of all horizontal edges $((x,y),(x+1,y))$ is 1, and the length of each vertical edge $((x,y),(x,y+1))$ is $1+\frac{1}{2^{x^2} \cdot k}$.
Let $R_1=\{(0,y):\space y\in [0, \frac{k}{2}-1]\}$, and $R_2=\{(x, x):\space x\in [\frac{k}{2},k-1]\}$. Let the terminals in the graph be $R=R_1\cup R_2$, so $|R|=k$. See Figure \ref{fig:grid} for illustration.

\begin{figure}[h]
\centering
\includegraphics[width=4 in]{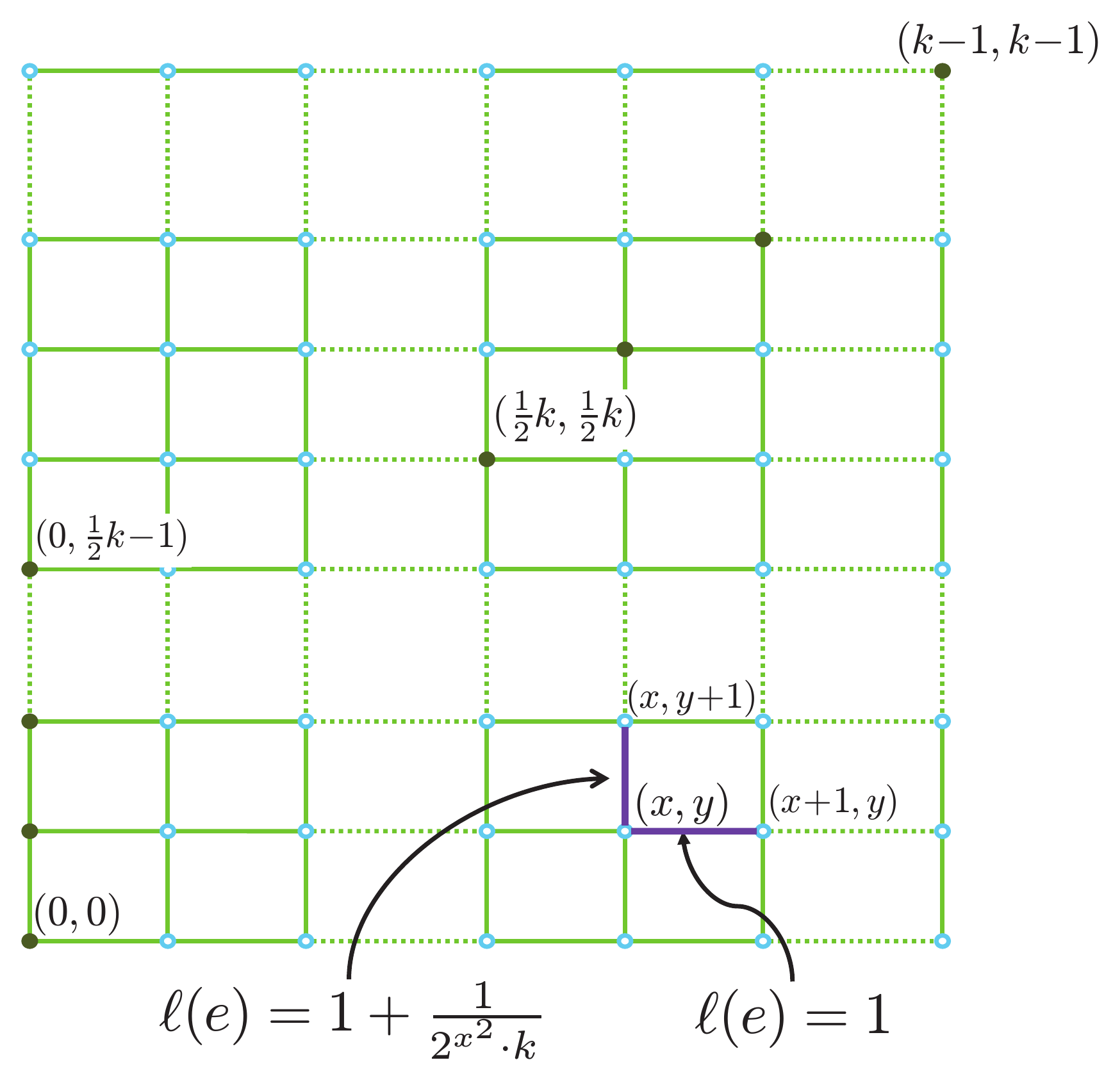}
\caption{A grid graph $G$ and terminals $R$.}
\label{fig:grid}
\end{figure}

It is easy to see that the shortest-path between a vertex $(0,y)\in R_1$ and a vertex $(x,x)\in R_2$ includes exactly $x$ horizontal edges and $x-y$ vertical edges. Indeed, such paths have length smaller than $x+(x-y)(1+\frac{1}{k})\le 2x-y+1$. Any other path between these vertices will have length greater than $2x-y+2$. Furthermore, the shortest path with $x$ horizontal edges and $x-y$ vertical edges starting at vertex $(0,y)$ makes horizontal steps before vertical steps, since the vertical edge-lengths decrease as $x$ increases, hence
\begin{equation}\label{eqn:distance}
d_G((0,y),(x,x)) = 2x-y+\frac{x-y}{2^{x^2} \cdot k}.
\end{equation}

Assume towards contradiction that there exists a planar graph $G'$ with less than $\frac{k^2}{1600}$ vertices that preserves terminal distances exactly. Since $G'$ is planar, by the weighted version of the planar separator theorem by Lipton and Tarjan \cite{LT79} with vertex-weight $1$ on terminals and 0 on non-terminals, there exists a partitioning of $V'$ into three sets $A_1$, $S$, and $A_2$ such that $w(S)\le |S|\le 2.5\cdot \sqrt{\frac{k^2}{1600}}< \frac{3k}{40}$, each of $A_1$ and $A_2$ has at most $\frac{2k}{3}$ terminals, and there are no edges going between $A_1$ and $A_2$. Hence, for $i\in \{1,2\}$ it holds that $w(A_i\cup S)\ge k/3$ and $w(A_i)\ge \frac{k}{3}-\frac{3k}{40}>\frac{k}{4}$.

Without loss of generality, we claim that $A_1\cap R_1$ and $A_2\cap R_2$ each have $\Theta(k)$ terminals. To see this, suppose without loss of generality that $A_1$ is the heavier of the two sets (i.e. $w(A_1)\ge \frac{k}{2}-\frac{3k}{40}$ and $\frac{k}{4} \le w(A_2)\le \frac{k}{2})$. Suppose also that $w(A_2\cap R_2)\ge w(A_2\cap R_1)$. Then $w(A_2\cap R_2)\ge \frac{k}{8}$, and $w(A_2\cap R_1)\le \frac{1}{2}\cdot w(A_2)\le \frac{k}{4}$, implying that $w(A_1\cap R_1)\ge w(R_1)-(w(R_1\cap A_2)+w(R_1\cap S))\ge \frac{k}{2}-(\frac{k}{4}+\frac{3k}{40}) = \frac{k}{5}$. In conclusion, without loss of generality it holds that $w(A_1\cap R_1)\ge \frac{k}{5}$ and $w(A_2\cap R_2)\ge \frac{k}{8}$. Let $Q_1\subseteq A_1\cap R_1$ and $Q_2\subseteq A_2\cap R_2$ be two sets with the exact sizes $\frac{k}{5}$ and $\frac{k}{8}$.

Every path between a terminal in $Q_1$ and a terminal in $Q_2$ goes through at least one vertex of the separator $S$. Overall, the vertices in the separator participate in $\frac{k}{8}\times \frac{k}{5}$ paths between $Q_1$ and $Q_2$. See Figure \ref{fig:sep} for illustration.

\begin{figure}[h]
\centering
\includegraphics[width=0.5\linewidth]{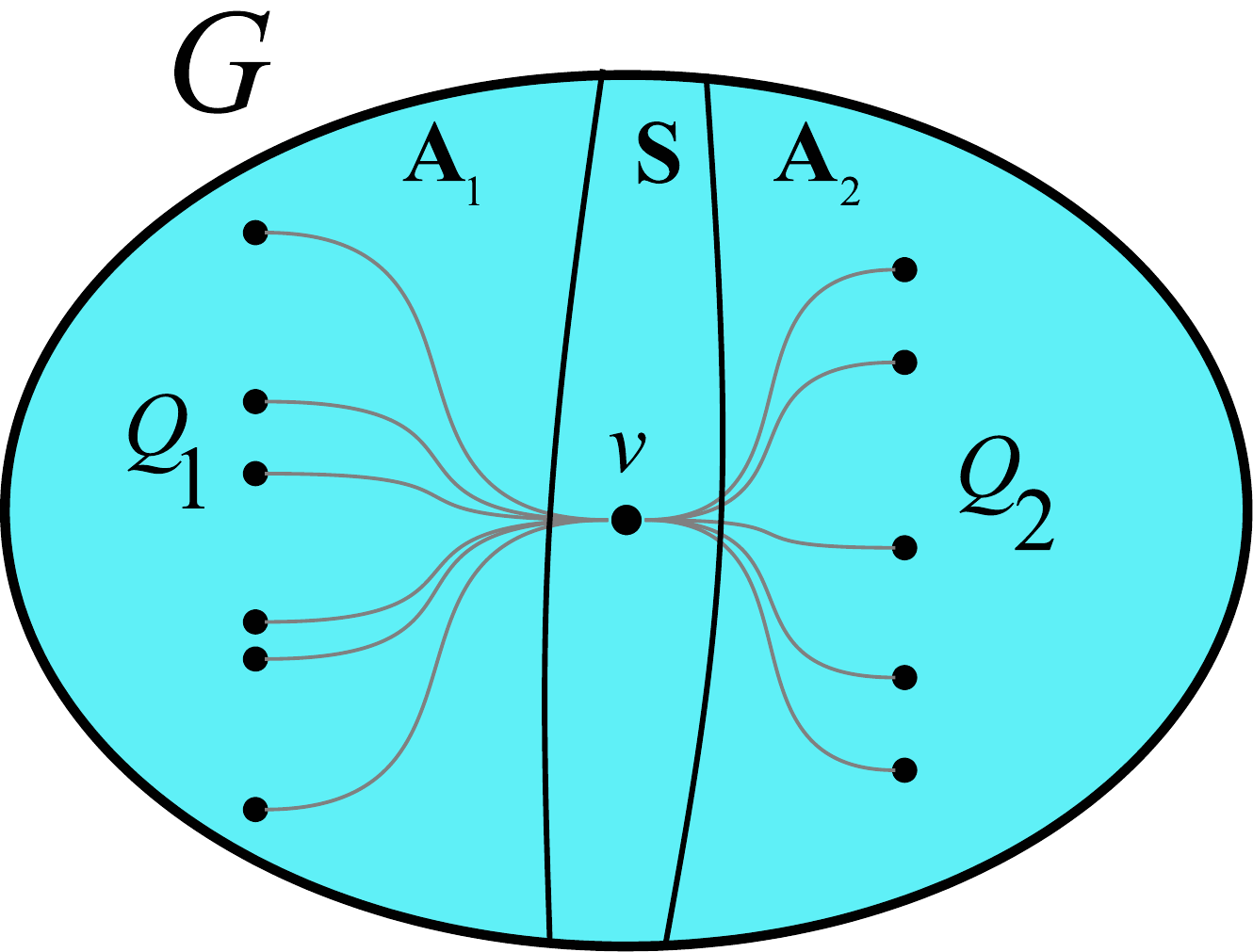}
\caption{Terminals on different sides connected by
paths going through $v\in S$.}
\label{fig:sep}
\end{figure}

We will need the following lemma, which is proved below.
\begin{lemma} \label{lem:thruv}
Let $G'$, $S$, $Q_1$ and $Q_2$ be as described above. Then every vertex $v\in S$ participates in at most $|Q_1|+|Q_2|=\frac{k}{5}+\frac{k}{8}$ shortest paths between $Q_1$ and $Q_2$.
\end{lemma}

Applying Lemma \ref{lem:thruv} to every vertex in $S$, at most $\frac{3k}{40}\cdot \frac{13k}{40} = \frac{39 k^2}{1600}<\frac{k^2}{40}$ shortest paths between $Q_1$ and $Q_2$ go through $S$, which contradicts the fact that all $\frac{k}{8}\cdot \frac{k}{5} = \frac{k^2}{40}$ shortest-paths between $Q_1$ and $Q_2$ in $G'$ go through the separator, and proves Theorem \ref{thm:grid}.  
\ifprocs\qed\fi
\end{proof}

\begin{proof}[\proofof{of Lemma \ref{lem:thruv}}]
Define a bipartite graph $H$ on the sets $Q_1$ and $Q_2$, with an edge between $(0,y)\in Q_1$ and $(x,x)\in Q_2$ whenever a shortest path in $G'$ between $(0,y)$ and $(x,x)$ uses the vertex $v$. We shall show that $H$ does not contain an even-length cycle. Since $H$ is bipartite, it contains no odd-length cycles either, making $H$ a forest with $|E(H)|<|Q_1|+|Q_2| = \frac{k}{5}+\frac{k}{8}$, thereby proving the lemma.

Let us consider a potential $2s$-length (simple) cycle in $H$ on the vertices $(0,y_1)$, $(x_1, x_1)$, $(0, y_2)$, $(x_2, x_2)$, ..., $(0,y_s)$, $(x_s, x_s)$ (in that order), for particular $(0, y_i)\in Q_1$ and $(x_i, x_i)\in Q_2$. 
Every edge $((0,y),(x,x))\in E(H)$ represents a shortest path in $G'$ that uses $v$, thus 
\begin{equation} \label{eqn:splitatv} d_G((0,y),(x, x)) = d_{G'}((0, y),v)+ d_{G'}(v,(x, x)).\end{equation}
If the above cycle exists in $H$, then the following equalities hold (by convention, let $y_{s+1} = y_1$). Essentially, we get that the sum of distances corresponding to ``odd-numbered" edges in the cycle equals the one corresponding to ``even-numbered" edges in the cycle.
\begin{eqnarray*} 
\sum_{i=1}^{s} d_{G}((0, y_i), (x_i, x_i)) &\overset{\text{(\ref{eqn:splitatv})}}{=}&
 \sum_{i=1}^{s} d_{G'}((0, y_i), v) + \sum_{i=1}^{s} d_{G'}(v, (x_i, x_i)) \\
&=& \sum_{i=1}^{s} d_{G'}(v, (0, y_{i+1}))+\sum_{i=1}^{s} d_{G'}((x_i, x_i), v) \\
&\overset{\text{(\ref{eqn:splitatv})}}{=}&
\sum_{i=1}^{s} d_{G}((x_i, x_i), (0, y_{i+1})).
\end{eqnarray*}

Plugging in the distances as described in (\ref{eqn:distance}) and simplifying, we obtain
\begin{equation*}
\sum_{i=1}^{s} (2x_i-y_i+(x_i-y_i)\cdot \frac{1}{2^{x_i^2} \cdot k}) = \sum_{i=1}^{s} (2x_i-y_{i+1}+ (x_i-y_{i+1})\cdot\frac{1}{2^{x_i^2} \cdot k}),
\end{equation*}
or equivalently, 
\begin{equation*}
\sum_{i=1}^{s}  \frac{y_i}{2^{x_i^2}}= \sum_{i=1}^{s}  \frac{y_{i+1}}{2^{x_i^2}}
\end{equation*}

Suppose without loss of generality that $x_1=\min\{x_i: i\in [1,s] \}$ (otherwise we can rotate the notations along the cycle), and that $y_1>y_2$ (otherwise we can change the orientation of the cycle). Then we obtain
$$\frac{y_1-y_2}{2^{x_1^2}} = \sum_{i=2}^{s}  \frac{y_{i+1}-y_i}{2^{x_i^2}}.$$
However, since $y_1> y_2$, the lefthand side is at least $\frac{1}{2^{x_1^2}}$, whereas the righthand side is
$\sum_{i=2}^{s} \frac{y_{i+1}-y_i}{2^{x_i^2}} \le (s-1)\cdot \frac{k}{2^{(x_1+1)^2}}\le \frac{k^2}{2^{(x_1+1)^2}}$. Therefore it must hold that $2^{2x_1+1}\le k^2.$
Since $x_1 \geq \frac{k}{2}$, this inequality does not hold. Hence, for all $s$, no cycle of size $2s$ exists in $H$, completing the proof of Lemma \ref{lem:thruv}. 

\ifprocs\qed\fi
\end{proof}

\section{$\Theta(k)$ Bounds for Constant Treewidth Graphs} \label{sec:tw}
In this section we prove Theorem \ref{thm:tw}, 
which bounds $f^*(k, \tw(p))$.
The upper and the lower bound are proved 
separately in Theorems \ref{thm:twupper} and \ref{thm:omegakh} below. 

\subsection{An Upper Bound of $O(p^3k)$}

\begin{theorem} \label{thm:twupper}
Every graph $G=(V,E,\ell)$ with treewidth $p$ and a set $R\subseteq V$ of $k$ terminals admits a distance-preserving minor $G'=(V', E', \ell ')$ with $|V'|\le O(p^3k)$.
In other words, $f^*(k, \tw (p))\le O(p^3k)$.
\end{theorem}
The graph $G'$ can in fact be computed in time polynomial in $|V|$ (see Remark \ref{rem:trivial}).

Without loss of generality, we may assume that $k\ge p$, since otherwise the $O(k^4)$ bound from Theorem \ref{thm:k4} applies. To prove Theorem \ref{thm:twupper} we introduce the algorithm \redtw (depicted in Algorithm \ref{alg:twcons} below), which follows a divide-and-conquer approach.
We use the small separators guaranteed by the treewidth $p$, 
to break the graph recursively until we have small, almost-disjoint subgraphs. 
We apply the naive algorithm 
(\naive, depicted in Algorithm \ref{alg:naive} in Section \ref{sec:trivial}) on each of these subgraphs with an altered set of terminals -- the original terminals in the subgraph, plus the separator (\emph{boundary}) vertices which disconnect these terminals from the rest of the graph.
we get many small distance-preserving minors, which are then combined into 
a distance-preserving minor $G'$ of the original graph $G$.

\begin{proof} [\proofof{of Theorem \ref{thm:twupper}}]
The divide-and-conquer technique works as follows. Given a partitioning of $V$ into the sets $A_1$, $S$ and $A_2$, such that removing $S$ disconnects $A_1$ from $A_2$, the graph $G$ is divided into the two subgraphs $G[A_i\cup S]$ (the subgraph of $G$ induced on $A_i\cup S$) for $i\in \{1,2\}$. 
For each $G[A_i\cup S]$, we compute a distance-preserving minor 
with respect to terminals set $(R\cap A_i)\cup S$,
and denote it $\hat G_i = (\hat V_i, \hat E_i, \hat \ell _i)$. 
The two minors are then combined into a distance-preserving minor of $G$ with respect to $R$, according to the following definition. 

We define the \emph{union} $H_1\cup H_2$ of two (not necessarily disjoint) 
graphs $H_1=(V_1, E_1, \ell_1)$ and $H_2=(V_2, E_2, \ell_2)$ to be the graph $H=(V_1\cup V_2,E_1\cup E_2,\ell)$ where the edge lengths are $\ell(e)=\min \{\ell_1(e), \ell_2(e)\}$ 
(assuming infinite length when $\ell_i(e)$ is undefined). 
A crucial point here is that $H_1,H_2$ need not be disjoint --
overlapping vertices are merged into one vertex in $H$,
and overlapping edges are merged into a single edge in $H$.

\begin{lemma} \label{lem:divncon}
The graph $\hat G=\hat G_1\cup \hat G_2$ is a distance-preserving minor of $G$ with respect to $R$.
\end{lemma}
\begin{proof} [\proofof{of Lemma \ref{lem:divncon}}]
Note that since the \emph{boundary vertices} in $S$ exist in both $\hat G_1$ and $\hat G_2$, they are never contracted into other vertices. In fact, the only  minor-operation allowed on vertices in $S$ is the removal of edges $(s_1, s_2)$ for two vertices $s_1, s_2\in S$, when shorter paths in $G[A_1\cup S]$ or $G[A_2\cup S]$ are found. It is thus possible to perform both sequences of minor-operations independently, making $\hat G$ a minor of $G$. 

A path between two vertices $t_1, t_2\in R$ can be split into subpaths at every visit to a vertex in $R\cup S$, so that each subpath between $v,u\in R\cup S$ does not contain any other vertices in $R\cup S$. Since there are no edges between $A_1$ and $A_2$, each of these subpaths exists completely inside $G[A_1\cup S]$ or $G[A_2\cup S]$. Hence, for every subpath between $v,u\in R\cup S$ it holds that $d_G(v,u) = d_{G[A_i\cup S]}(v,u) = d_{\hat G_i}(v,u)$ for some $i\in \{1,2\}$. Altogether, the shortest path in $G$ is preserved in $\hat G$. It is easy to see that shorter paths will never be created, as these too can be split into subpaths such that the length of each subpath is preserved. Hence, $\hat G$ is a distance-preserving minor of $G$.
\ifprocs\qed\fi
\end{proof}

The graph $G$ has bounded treewidth $p$, hence for every nonnegative vertex-weights $w(\cdot)$, there exists a set $S\subseteq V$ of at most $p+1$ vertices (to simplify the analysis, we assume this number is $p$) whose removal separates the graph into two parts $A_1$ and $A_2$, each with $w(A_i)\le \frac{2}{3} w(V)$.
It is then natural to compute a distance-preserving minor for each part $A_i$
by recursion,
and then combine the two solutions using Lemma \ref{lem:divncon}.
We can use the weights $w(\cdot)$ to obtain a balanced split of the 
terminals, and thus $|R\cap A_i|$ is a constant factor smaller than $|R|$. 
However, when solving each part $A_i$, the boundary vertices $S$ 
must be counted as ``additional'' terminals,
and to prevent those from accumulating too rapidly,
we compute (\ala \cite{Bodlaender89})
a second separator $S^i$ with different weights $w(\cdot)$ 
to obtain a balanced split of the boundary vertices accumulated so far.

Algorithm \redtw receives, in addition to a graph $H$ and a set of terminals $R\subseteq V(H)$, a set of boundary vertices $B\subseteq V(H)$. Note that a terminal that is also on the boundary is counted only in $B$ and not in $R$, so that $R\cap B = \emptyset$.

The procedure $\separator(H, U)$ returns the triple $\langle A_1, S, A_2 \rangle$ of a separator $S$ and two sets $A_1$ and $A_2$ such that $|S|\le p$, no edges between $A_1$ and $A_2$ exist in $G$, and $|A_1\cap U|, |A_2\cap U|\le \frac{2}{3}|U|$,
i.e., using $w(\cdot)$ that is unit-weight inside $U$ and $0$ otherwise.
\algsetup{indent=2em}
\begin{algorithm}                      
\caption{\redtw(graph $H$, required vertices $R$, boundary vertices $B$)}          
\label{alg:twcons}                           
\begin{algorithmic}[1]

\IF {$|R\cup B|\le 18p$}
\RETURN $\naive(H,R\cup B)$\quad (see Algorithm \ref{alg:naive})
\ENDIF
\STATE $\langle A_1,S,A_2 \rangle \gets$ \separator($H, R$)
\FOR {$i=1,2$}
\STATE $\langle A_i^1, S^i, A_i^2\rangle\gets \separator(H[A_i\cup S], (B\cap A_i)\cup S)$
\STATE $R^i\gets R\setminus (S\cup S^i)$
\STATE $B^i\gets B\cup S\cup S^i$
\FOR {$j=1,2$}
\STATE $\hat G_i^j\gets \redtw(H[A_i^j\cup S^i], R^i\cap A_i^j, B^i\cap (A_i^j\cup S^i))$
\ENDFOR
\ENDFOR

\RETURN $(\hat G_1^1\cup \hat G_1^2)\cup (\hat G_2^1\cup \hat G_2^2)$.

\end{algorithmic}
\end{algorithm}

See Figure \ref{fig:algsep} for an illustration of a single execution. Consider the recursion tree $T$ on this process, starting with the invocation of $\redtw(G, R, \emptyset)$. A node $a\in V(T)$ corresponds to an invocation $\redtw(H_a, R_a, B_a)$. The execution either terminates at line 2 (the stop condition), or performs 4 additional invocations $b_i$ for $i\in [1,4]$, each with $|R_{b_i}|\le \frac{2}{3}|R_a|$. As the process continues, the number of terminals in $R_a$ decreases, whereas the number of boundary vertices may increase. We show the following upper bound on the number of boundary vertices $B_a$.

\begin{figure}[htb]
\centering
\includegraphics[scale=0.8]{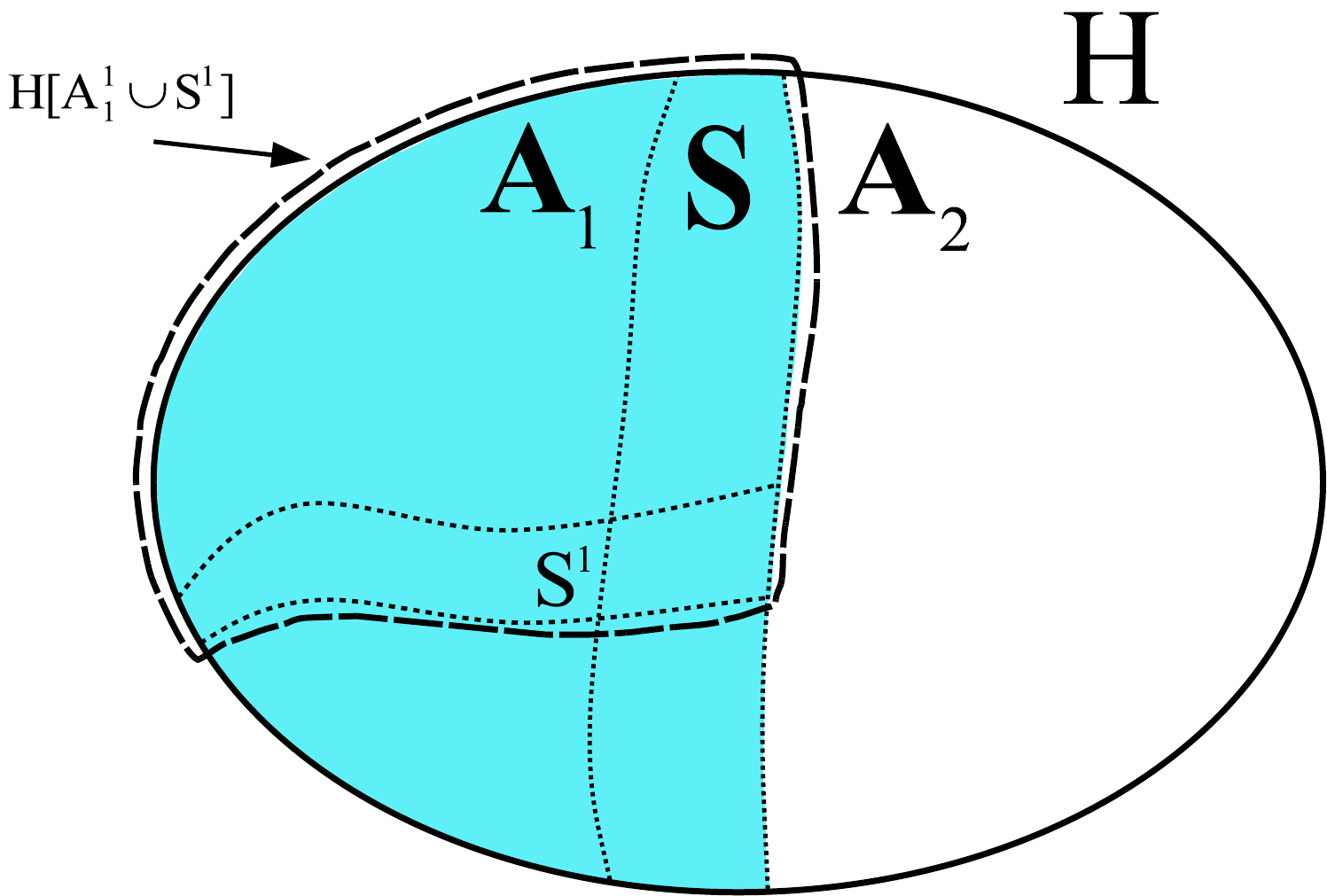} 
\caption{The separators $S$ (from line 3) and $S^1$ (from line 7),
and the subgraph $H[A_1^1\cup S^1]$ to be processed recursively (in line 11).}
\label{fig:algsep}
\end{figure}

\begin{lemma} \label{lem:W6TW}
For every $a\in V(T)$, the number of boundary vertices $|B_a|< 6p$. 
\end{lemma}

\begin{proof}[\proofof{of Lemma \ref{lem:W6TW}}]
Proceed by induction on the depth of the node in the recursion tree.
The lemma clearly holds for the root of the recursion-tree, since initially $B=\emptyset$. Suppose it holds for an execution with values $H_a$, $R_a$, $B_a$. When partitioning $V(H_a)$ into $A_1$, $S$, and $A_2$, the separator $S$ has at most $p$ vertices. From the induction hypothesis, $|B_a|<6p$, making $|B_a\cup S|< 7p$. 

The algorithm constructs another separator, this time separating the boundary vertices $B_a\cup S$. For $i=1,2$ and $j=1,2$ it holds that, $|S^i|\le p$,  $|A_i^j|\le \frac{2}{3}\cdot |B_a\cup S|\le \frac{2}{3}\cdot 7p = \frac{14}{3} p$, and so $|A_i^j\cup S^i|\le \frac{14}{3}p+p < 6p$. The execution corresponding to the node $a$ either terminates in line 2, or invokes executions with the values $A_i^j\cup S^i$ for $i,j=1,2$, hence all new invocations have less than $6p$ boundary vertices.
\ifprocs\qed\fi
\end{proof}

We also prove the following lower bound on the number of terminals $R_a$.

\begin{lemma} \label{lem:fullbin}
Every $a\in V(T)$ is either a leaf of the tree $T$, or it has at least two children, denoted $b_1, b_2$, such that $|R_{b_1}|, |R_{b_2}|\ge p$.
\end{lemma}

\begin{proof}[\proofof{of Lemma \ref{lem:fullbin}}]
Consider a node $a\in V(T)$. If this execution terminates at line 2, $a$ is a leaf and the lemma is true. Otherwise it holds that $|R_a\cup B_a|\ge 18p$. Since Lemma \ref{lem:W6TW} states that $|B_a|\le 6p$ it must holds that $|R_a|\ge 12p$.

When performing the separation of $V(H_a)$ into $A_1$, $S$, and $A_2$, the vertices $R_a$ are distributed between $A_1$, $S$, and $A_2$, such that $|R_a\cap (A_i\cup S)|\ge \frac{1}{3}|R_a|=4p$ for $i=1,2$. Since $|S|\le p$ it must holds that $|(R_a\setminus S) \cap A_i| = |(R_a\cap (A_i\cup S))\setminus S| \ge  3p$. 
When the next separation is performed, at most $p$ of these $3p$ terminals belong to $S^i$, while the remaining terminals belong to $R^i$ and are distributed between $A_i^1$ and $A_i^2$. At least one of these sets, without loss of generality $A_i^1$, gets $|R^i\cap A_i^1|\ge \frac{1}{2}2p = p$. This is a value of $R_{b}$ for a child $b$ of $a$ in the recursion tree. Since this holds for both $A_1$ and $A_2$, at least two invocations $b_1, b_2$ with $|R_{b_i}|\ge p$ are made.
\ifprocs\qed\fi
\end{proof}

The following observation is immediate from Lemma \ref{lem:W6TW}.
\begin{observation} \label{obs:leaf}
Every node $a\in V(T)$ such that $|R_a|<p$ has $|R_a\cup B_a|\le 7p$, thus it is a leaf in $T$.
\end{observation}

To bound the size of the overall combined graph $G'$ returned by the first call to $\redtw$, we must bound the number of leaves in $T$. To do that, we first consider the recursion tree $T'$ created by removing those nodes $a$ with $|R_a|<p$; these are leaves from Observation \ref{obs:leaf}. From Lemma \ref{lem:fullbin} every node in this tree (except the root) is either a leaf (with degree 1) or has at least two children (with degree at least 3). Since the average degree in a tree is less than 2, the number of nodes with degree at least 3 is bounded by the number of leaves. Every leaf $b$ in the tree $T'$ has $|R_b|\ge p$. These terminals do not belong to any boundary, so for every other leaf $b'$ in $T'$ it holds that $R_b\cap (R_{b'}\cup B_{b'})=\emptyset$ and these $p$ terminals are unique. There are $k$ terminals in $G$, so there are $O(k/p)$ such leaves, and $O(k/p)$ internal nodes.

From Lemma \ref{lem:fullbin}, invocations are performed only by by internal vertices in $T'$. Each internal vertex has 4 children, hence there are $O(k/p)$ invocations overall. Each leaf in $T$ has $|R_a\cup B_a|\le O(p)$, hence the graph returned from $\naive(H_a)$ is a distance-preserving minor with  $O(p^4)$ vertices (see section \ref{sec:trivial}). Using Lemma \ref{lem:divncon}, the combination of these graphs is a distance-preserving minor $\hat G$ of $G$ with respect to $R$. The minor $\hat G$ has $O(k/p\cdot p^4) = O(k\cdot p^3)$ vertices, proving Theorem \ref{thm:twupper}.
\ifprocs\qed\fi
\end{proof}

\begin{remark} \label{rem:trivial}
Every action (edge or vertex removals, as well as edge contractions) taken by \redtw, is actually performed during a call to \naive, and an equivalent action to it would have been taken had we executed the naive algorithm directly on $G$ with respect to terminals $R$. It follows that the naive algorithm, too, returns distance-preserving minors of size $O(k\cdot p^3)$ to any graph with treewidth $p$. 
(When $p>k$ this statement holds by the $O(k^4)$ bound.) 
\end{remark}

\subsection{A Lower Bound of $\Omega(pk)$}

\begin{theorem} \label{thm:omegakh}
For every $p$ and $k\ge p$ there is a graph $G=(V,E,\ell)$ with treewidth $p$ and $k$ terminals $R\subseteq V$, such that every distance-preserving minor $G'$ of $G$ with respect to $R$ has $|V'|\ge \Omega(k\cdot p)$. 
In other words, $f^*(k, \tw (p))\ge \Omega(pk)$.
\end{theorem}

\begin{proof}
Consider the bound shown in Theorem \ref{thm:grid}. The graph used to obtain this bound is a $k\times k$ grid, and has treewidth $k$. The following corollary holds.
\begin{corollary} \label{cor:h2}
For every $p\in \mathbb{N}$ there exists a graph $G$ with treewidth $p$ and $p$ terminals $R\subseteq V$, such that every distance-preserving minor $G'$ of $G$ with respect to $R$ has $|V'|\ge \Omega(p^2)$. 
\end{corollary}

Let the graph $G$ consist of $\frac{k}{p}$ disjoint graphs $G_i$ with $p$ terminals, treewidth $p$, and distance-preserving minors with $|V'|\ge \Omega(p^2)$ as guaranteed by Corollary \ref{cor:h2}. Any distance-preserving minor of the graph $G$ must preserve (in disjoint components) the distances between the terminals in each $G_i$. The graph $G$ has $k$ terminals, treewidth $p$, and any distance-preserving minor of it has $|V'|\ge \Omega(k\cdot p)$, thus proving Theorem \ref{thm:omegakh}.
\ifprocs\qed\fi
\end{proof}

\section{Minors with Dominating Distances} \label{sec:ConcRem}

The algorithms mentioned in this paper (including the naive one)
actually satisfy a stronger property: 
They output a minor $G'=(V',E',\ell')$ where in effect $V'\subseteq V$
(every vertex in $G'$ can be mapped back to a vertex in $G$) and
distances in $G'$ dominate those in $G$, namely
\begin{equation} \label{eqn:dominating}
  d_{G'}(u,v)\ge d_G(u,v) \qquad \forall u,v\in V'.
\end{equation}
The following theorem proves under this stronger property,
the $O(k^4)$ bound of Theorem~\ref{thm:k4} is tight.

\begin{theorem} \label{thm:dominating}
For every $k$ there exists a graph $G$ and a set of terminals $R\subseteq V$, for which every distance-preserving minor $G'$ where $V'\subseteq V$ and property \eqref{eqn:dominating} holds, has $\Omega(k^4)$ vertices. 
\end{theorem}

\begin{proof} Fix $k$; we construct $G$ probabilistically as follows. Consider the unit square $[0,1]\times [0,1]$ in the 2-dimensional Euclidean plane, and on each of its edges place terminals at $\lfloor \frac{k}{4} \rfloor$ points chosen at random. Connect by a straight line the terminals on the top edge with those on the bottom edge, and similarly connect the terminals on the right edge with those on the left edge. There are now $\Theta(k^2)$ ``horizontal'' lines each meeting $\Theta(k^2)$ ``vertical'' lines, and with probability $1$ the horizontal lines intersect the vertical lines at $\Theta(k^4)$ intersection points (because the probability that three lines meet at a single point is $0$). Additional intersection points might exist between pairs of horizontal lines and pairs of vertical lines. 

Let the graph $G$ have both the terminals and the intersection points as its vertices, and their connecting line segments as its edges. Set every edge length to be the Euclidean distance between its endpoints, hence shortest-path distances in $G$ dominate the Euclidean metric between the respective points.

Let $v$ be an intersection point between the top-to-botoom (horizontal) shortest-path $\Pi_G(t_1, t_2)$ and the right-to-left (vertical) shortest-path $\Pi_G(t_3, t_4)$ in $G$. Let $G'$ be a distance-preserving minor of $G$ satisfying property \eqref{eqn:dominating} and assume towards contradiction that $v\notin V'$. It is easy to see that $G'$ can be drawn in the 2-dimensional Euclidean plane in such a way that the surviving vertices and edges remain in the same location, and new edges are drawn inside the unit square. Since every pair of top-to-bottom path and right-to-left path (both inside the unit square) must intersect, the shortest-paths $\Pi_{G'}(t_1, t_2)$ and $\Pi_{G'}(t_3, t_4)$ intersect in some point $v'\in V'$, which must be different from $v$ (because $v\notin V'$). 
But since $v$ is the only vertex in $V\supset V'$ placed on both the straight line between $t_1$ and $t_2$, and the straight line between $t_3$ and $t_4$, one of the paths in $G'$, say without loss of generality $\Pi_{G'}(t_1, t_2)$, visits the point $v'$ and goes outside of its straight line. From property \eqref{eqn:dominating} all distances  in $G'$ dominate those in $G$, and from the construction of $G$ they also dominate the Euclidean metric. Hence, the length of the shortest-path $\Pi_{G'}(t_1, t_2)$ is at least the sum of Euclidean distances $\|t_1-v'\|_2+\|v'-t_2\|_2 > \|t_1-t_2\|_2$, making $d_{G'}(t_1, t_2) > d_G(t_1, t_2)$ in contradiction to the distance-preserving property of $G'$. We conclude that every intersection point between a vertical and a horizontal line in $G$ exists also in $G'$, hence $|V'| \ge \Omega(k^4)$. 
\end{proof}

Theorem \ref{thm:dominating} suggests that narrowing the gap between the current bounds 
$\Omega(k^2) \le f^*(k)\le O(k^4)$,
might require, even for planar graphs, breaking away from the above paradigm
and not satisfying property \eqref{eqn:dominating}.

{ 
\bibliographystyle{alphainit}
\bibliography{biblio}

\newcommand{\etalchar}[1]{$^{#1}$}
\begin{thebibliography}{CGN{\etalchar{+}}06}

\bibitem[BG08]{BG08}
A.~Basu and A.~Gupta.
\newblock Steiner point removal in graph metrics.
\newblock Unpublished Manuscript, available from
  \url{http://www.math.ucdavis.edu/~abasu/papers/SPR.pdf}, 2008.

\bibitem[Bod89]{Bodlaender89}
H.~L. Bodlaender.
\newblock {NC}-algorithms for graphs with small treewidth.
\newblock In {\em 14th International Workshop on Graph-Theoretic Concepts in
  Computer Science}, pages 1--10. Springer-Verlag, 1989.

\bibitem[CE06]{CE06}
D.~Coppersmith and M.~Elkin.
\newblock Sparse sourcewise and pairwise distance preservers.
\newblock {\em SIAM J. Discrete Math.}, 20:463--501, 2006.

\bibitem[CGN{\etalchar{+}}06]{CGNRS06}
C.~Chekuri, A.~Gupta, I.~Newman, Y.~Rabinovich, and A.~Sinclair.
\newblock Embedding $k$-outerplanar graphs into {$\ell_1$}.
\newblock {\em SIAM J. Discret. Math.}, 20(1):119--136, 2006.

\bibitem[CXKR06]{CXKR06}
T.~Chan, D.~Xia, G.~Konjevod, and A.~Richa.
\newblock A tight lower bound for the {S}teiner point removal problem on trees.
\newblock In {\em 9th International Workshop on Approximation, Randomization,
  and Combinatorial Optimization}, volume 4110 of {\em Lecture Notes in
  Computer Science}, pages 70--81. Springer, 2006.

\bibitem[DHZ00]{DHZ00}
D.~Dor, S.~Halperin, and U.~Zwick.
\newblock All-pairs almost shortest paths.
\newblock {\em SIAM J. Comput.}, 29(5):1740--1759, 2000.

\bibitem[EGK{\etalchar{+}}10]{EGKRTT10}
M.~Englert, A.~Gupta, R.~Krauthgamer, H.~R{\"a}cke, I.~Talgam-Cohen, and
  K.~Talwar.
\newblock Vertex sparsifiers: New results from old techniques.
\newblock In {\em 13th International Workshop on Approximation, Randomization,
  and Combinatorial Optimization}, volume 6302 of {\em Lecture Notes in
  Computer Science}, pages 152--165. Springer, 2010.

\bibitem[FM95]{FM95}
T.~Feder and R.~Motwani.
\newblock Clique partitions, graph compression and speeding-up algorithms.
\newblock {\em J. Comput. Syst. Sci.}, 51(2):261--272, 1995.

\bibitem[GH61]{GH61}
R.~E. Gomory and T.~C. Hu.
\newblock Multi-terminal network flows.
\newblock {\em Journal of the Society for Industrial and Applied Mathematics},
  9:551--570, 1961.

\bibitem[GS02]{GS02}
M.~Grigni and P.~Sissokho.
\newblock Light spanners and approximate {TSP} in weighted graphs with
  forbidden minors.
\newblock In {\em 13th Annual ACM-SIAM Symposium on Discrete Algorithms}, pages
  852--857. SIAM, 2002.

\bibitem[Gup01]{Gupta01}
A.~Gupta.
\newblock Steiner points in tree metrics don't (really) help.
\newblock In {\em 12th Annual ACM-SIAM Symposium on Discrete Algorithms}, pages
  220--227. SIAM, 2001.

\bibitem[HKNR98]{HKNR98}
T.~Hagerup, J.~Katajainen, N.~Nishimura, and P.~Ragde.
\newblock Characterizing multiterminal flow networks and computing flows in
  networks of small treewidth.
\newblock {\em J. Comput. Syst. Sci.}, 57:366--375, 1998.

\bibitem[IS07]{IS07}
P.~Indyk and A.~Sidiropoulos.
\newblock Probabilistic embeddings of bounded genus graphs into planar graphs.
\newblock In {\em 23rd Annual Symposium on Computational Geometry}, pages
  204--209. ACM, 2007.

\bibitem[Kle08]{Klein08}
P.~N. Klein.
\newblock A linear-time approximation scheme for {TSP} in undirected planar
  graphs with edge-weights.
\newblock {\em SIAM J. Comput.}, 37(6):1926--1952, 2008.

\bibitem[LT79]{LT79}
R.~J. Lipton and R.~E. Tarjan.
\newblock A separator theorem for planar graphs.
\newblock {\em SIAM J. Appl. Math.}, 36(2):177--189, 1979.

\bibitem[PS89]{PS89}
D.~Peleg and A.~A. Sch{\"a}ffer.
\newblock Graph spanners.
\newblock {\em J. Graph Theory}, 13(1):99--116, 1989.

\bibitem[TZ05]{TZ05}
M.~Thorup and U.~Zwick.
\newblock Approximate distance oracles.
\newblock {\em J. ACM}, 52(1):1--24, 2005.

\bibitem[Woo06]{Woodruff06}
D.~P. Woodruff.
\newblock Lower bounds for additive spanners, emulators, and more.
\newblock In {\em 47th Annual IEEE Symposium on Foundations of Computer
  Science}, pages 389--398. IEEE Computer Society, 2006.

\end{thebibliography}
}

\end{document}